\ifx\synctex\undefined\else\synctex=1\fi
\pdfoutput=1
\documentclass{tlp}

\usepackage{color}
\usepackage{graphicx}

\usepackage[fleqn]{amsmath} 
\usepackage{algorithm}
\usepackage[noend]{algpseudocode}
\usepackage{alltt}
\usepackage{tikz}
\usetikzlibrary{shapes}
\usepackage{adjustbox} 
\usepackage{multicol}

\usepackage{amssymb}
\usepackage{upquote}
\usepackage{fancyvrb} %
\usepackage[belowskip=-15pt,aboveskip=5pt]{caption} %
\usepackage{float}
\usepackage{wrapfig}
\usepackage[all]{xy}

\usepackage{hyperref}

\usepackage{upquote} %
\usepackage{listings}%
\usepackage{enumerate}

\newcommand{\ignore}[1]{}

\newcommand{\pcode}[2][\codesize]{
    \fbox{
    \begin{minipage}{0.45\linewidth}
    #1
    \begin{tabbing}
    xx \= xx \= xx \= xx \= xx \= xx \= xx \= xx \= xx \= xx \= \kill
    #2
    \end{tabbing}
    \end{minipage}
    }
  }

\title[An iterative approach to precondition inference using CHCs]
{An iterative approach to precondition inference using constrained Horn clauses}

\author[Kafle, Gallagher, Gange, Schachte, S{\o}ndergaard, Stuckey]
{Bishoksan Kafle\\
The University of Melbourne
\and
John P. Gallagher\\
Roskilde University and IMDEA Software Institute
\and 
Graeme Gange\\
Monash University
\and
Peter Schachte, Harald S{\o}ndergaard, Peter J. Stuckey\\
The University of Melbourne
}

\begin{document}
\maketitle

%
\newcommand{\rahft}{\textsc{Rahft}}
\newcommand{\integ}{{\sf int}}
\newcommand{\listint}{{\sf listint}}
\newcommand{\other}{{\sf other}}
\newcommand{\true}{\textsc{True}}
\newcommand{\false}{\textsc{False}}
\newcommand{\trueit}{\mathit{true}}
\newcommand{\falseit}{\mathit{false}}
\newcommand{\falsett}{\mathtt{false}}
\newcommand{\Bin}{{\sf Bin}}
\newcommand{\Dep}{{\sf Dep}}
\newcommand{\g}{{\sf g}}
\newcommand{\nong}{{\sf ng}}
\newcommand{\OL}{{\cal O}}
\newcommand{\M}{{\sf M}}
\newcommand{\R}{{\cal R}}
\newcommand{\A}{\mathcal{A}}

\newcommand{\body}{\mathcal{B}}
\newcommand{\B}{{\cal B}}
\newcommand{\C}{{\cal C}}
\newcommand{\D}{{\cal D}}
\newcommand{\X}{{\cal X}}
\newcommand{\V}{{\cal V}}
\newcommand{\Q}{{\cal Q}}
\newcommand{\F}{{\sf F}}
\newcommand{\N}{{\cal N}}
\newcommand{\Lang}{{\cal L}}
\newcommand{\powerset}{{\cal P}}
\newcommand{\FTA}{{\cal FT\!A}}
\newcommand{\Term}{{\sf Term}}
\newcommand{\Empty}{{\sf empty}}
\newcommand{\nonEmpty}{{\sf nonempty}}
\newcommand{\compl}{{\sf complement}}
\newcommand{\args}{{\sf args}}
\newcommand{\preds}{{\sf preds}}
\newcommand{\gnd}{{\sf gnd}}
\newcommand{\lfp}{{\sf lfp}}
\newcommand{\psharp}{P^{\sharp}}
\newcommand{\minimize}{{\sf minimize}}
\newcommand{\headterms}{\mathsf{headterms}}
\newcommand{\solvebody}{\mathsf{solvebody}}
\newcommand{\solve}{\mathsf{solve}}
\newcommand{\fail}{\mathsf{fail}}
\newcommand{\member}{\mathsf{memb}}
\newcommand{\ground}{\mathsf{ground}}
\newcommand{\abst}{\mathsf{abstract}}
\newcommand{\bodyfacts}{\mathsf{bodyfacts}}
\newcommand{\renameunfold}{\mathsf{renameunfold}}
\newcommand{\rep}{\mathsf{rep}}
\newcommand{\cfacts}{\mathsf{cfacts}}

\newcommand{\raf}{{\sf raf}}
\newcommand{\qa}{{\sf qa}}
\newcommand{\spl}{{\sf split}}

\newcommand{\transitions}{\mathsf{transitions}}
\newcommand{\nonempty}{\mathsf{nonempty}}
\newcommand{\dom}{\mathsf{dom}}

\newcommand{\Args}{\mathsf{Args}}
\newcommand{\id}{\mathsf{id}}
\newcommand{\type}{\tau}
\newcommand{\restrict}{\mathsf{restrict}}
\newcommand{\any}{\top}
\newcommand{\dyn}{\top}
\newcommand{\dettypes}{{\sf dettypes}}
\newcommand{\Atom}{{\sf Atom}}

\newcommand{\chc}{{\sf chc}}
\newcommand{\deriv}{{\sf deriv}}

\newcommand{\vars}{\mathsf{vars}}
\newcommand{\Vars}{\mathsf{Vars}}
\newcommand{\range}{\mathsf{range}}
\newcommand{\varpos}{\mathsf{varpos}}
\newcommand{\varid}{\mathsf{varid}}
\newcommand{\argpos}{\mathsf{argpos}}
\newcommand{\elim}{\mathsf{elim}}
\newcommand{\pred}{\mathsf{pred}}
\newcommand{\predfuncs}{\mathsf{predfuncs}}
\newcommand{\project}{\mathsf{project}}
\newcommand{\reduce}{\mathsf{reduce}}
\newcommand{\positions}{\mathsf{positions}}
\newcommand{\contained}{\preceq}
\newcommand{\equivalent}{\cong}
\newcommand{\unify}{{\it unify}}
\newcommand{\Iff}{{\rm iff}}
\newcommand{\Where}{{\rm where}}
\newcommand{\qmap}{{\sf qmap}}
\newcommand{\fmap}{{\sf fmap}}
\newcommand{\ftable}{{\sf ftable}}
\newcommand{\Qmap}{{\sf Qmap}}
\newcommand{\states}{{\sf states}}
\newcommand{\head}{\tau}
\newcommand{\atomconstraints}{\mathsf{atomconstraints}}
\newcommand{\thresholds}{\mathsf{thresholds}}
\newcommand{\term}{\mathsf{Term}}
\newcommand{\trees}{\mathsf{trees}}
\newcommand{\renames}{\rho}
\newcommand{\renameps}{\rho_2}
\newcommand{\predicates}{\mathsf{Predicates}}
\newcommand{\query}{\mathsf{q}}
\newcommand{\ans}{\mathsf{a}}
\newcommand{\trace}{\mathsf{tr}}
\newcommand{\constr}{\mathsf{constr}}
\newcommand{\Iproj}{\mathsf{proj}}
\newcommand{\SAT}{\mathsf{SAT}}
\newcommand{\interpolant}{\mathsf{interpolant}}
\newcommand{\unknown}{?}
\newcommand{\rhs}{{\sf rhs}}
\newcommand{\lhs}{{\sf lhs}}
\newcommand{\unfold}{{\sf unfold}}
\newcommand{\arity}{{\sf ar}}
\newcommand{\AND}{{\sf AND}}

\newcommand{\feasible}{{\sf feasible}}
\newcommand{\infeasible}{{\sf infeasible}}
\newcommand{\safe}{{\sf safe}}
\newcommand{\unsafe}{{\sf unsafe}}
\newcommand{\swp}{{\sf presafe}} %
\newcommand{\spec}{{\sf spec}} %
\newcommand{\theory}{\mathbb{T}}

\newcommand{\atmost}[1]{\le #1}
\newcommand{\exactly}[1]{=#1}
\newcommand{\exceeds}[1]{>#1}
\newcommand{\anydim}[1]{\ge 0}

\def\ll{[\![}
\def\rr{]\!]}

\newcommand{\sset}[2]{\left\{~#1  \left|
                               \begin{array}{l}#2\end{array}
                          \right.     \right\}}

\newcommand{\qin}{\hspace*{0.15in}}
\newenvironment{SProg}
     {\begin{small}\begin{tt}\begin{tabular}[t]{l}}%
     {\end{tabular}\end{tt}\end{small}}
\def\anno#1{{\ooalign{\hfil\raise.07ex\hbox{\small{\rm #1}}\hfil%
        \crcr\mathhexbox20D}}}

\newtheorem{definition}{Definition}
\newtheorem{example}{Example}
\newtheorem{corollary}{Corollary}

\newtheorem{lemma}{Lemma}
\newtheorem{theorem}{Theorem}
\newtheorem{proposition}{Proposition}
\newtheorem{property}{Property}

\newcommand{\tuplevar}[1]{\mathbf{#1}}
 
\begin{abstract}
    %
We present a method for automatic inference of conditions on the initial states of a program that guarantee that the safety assertions in the program are not violated. Constrained Horn clauses (CHCs) are used to model the program and assertions in a uniform way, and we use standard abstract interpretations to derive an over-approximation of the set of \emph{unsafe} initial states. The precondition then is the constraint corresponding to the complement of that set, under-approximating the set of \emph{safe} initial states. This idea of complementation is not new, but previous attempts to exploit it have suffered from loss of precision. Here we develop an iterative 
specialisation algorithm 
to give more precise, and in some cases optimal safety conditions.  
The algorithm combines existing transformations, namely constraint specialisation, partial evaluation and a trace elimination transformation. The last two of these transformations
perform polyvariant specialisation, leading to disjunctive constraints which improve precision.  
The algorithm is implemented and tested on a benchmark suite of programs from the literature in precondition inference and software verification competitions. 

Under consideration for acceptance in TPLP.
 \end{abstract}
  \begin{keywords}
    Precondition inference, backwards analysis,  abstract interpretation,  refinement, program specialisation, program transformation.
  \end{keywords}

%
\section{Introduction}\label{sec:intro}

Given a program with properties required to hold at specific
program points, \emph{precondition analysis} derives the conditions 
on the initial states ensuring that the properties hold.  
This has important applications in program verification, 
symbolic execution, program understanding and debugging.
While forward abstract interpretation approximates the set of reachable 
states of a program, \emph{backward} abstract interpretation
approximates the set of states that can reach some target state.  
Both forward and backward analyses may produce 
over- or under-approximations, and
forward and backward analysis may profitably be combined
\cite{CousotCousot92-3,CCL_VMCAI11,DBLP:conf/sas/BakhirkinM17}.

Most approaches that apply backward analysis, possibly in conjunction
with forward analysis, use over-approximations, and as a result
derive \emph{necessary} pre-conditions.
Less attention has been given to under-approximating
backwards analyses, with the goal of finding \emph{sufficient} 
pre-conditions.
However, it is natural to try to derive
guarantees of safe behaviour of a program.
Often we would like to know which initial states \emph{must} be safe,
in the sense that no computation starting from such a state
can possibly reach a specified error state, that is,
we desire to find (non-trivial) sufficient conditions for safety.  

If analysis uses an abstract domain which is \emph{complemented},
duality enables
sufficient conditions to be derived from necessary conditions and 
\emph{vice versa}.
However, complemented abstract domains are very rare,
and approximation of a complement tends to introduce
considerable lack of precision.
The under-approximating backward abstract interpretation 
of \citeN{DBLP:conf/lopstr/HoweKL04} utilises the fact that the abstract
domain \emph{Pos} is pseudo-complemented~\cite{Marriott-Sondergaard-LOPLAS93},
but pseudo-complementation too is very rare.
\citeN{DBLP:conf/vmcai/Moy08} presents a method for deriving 
sufficient preconditions (for use with a theorem prover), 
employing weakest-precondition
reasoning and forward abstract interpretation to attempt to generalise
conditions at loop heads.
\citeN{DBLP:conf/sas/BakhirkinBP14} observe that there may be an
advantage in generalising an abstract complement operation to 
(abstract) logical subtraction, as this can improve opportunities
to find a tighter approximation of a set of states. 

\citeN{DBLP:journals/entcs/Mine12} infers sufficient conditions for safety,
not by instantiating a generic mechanism for complementation, but by
designing all required purpose-built backward transfer functions.
He does this for three numeric
abstract domains: intervals, octagons and convex polyhedra---a
substantial effort, as the purpose-built operations, 
including widening, can be rather intricate.

We share Min\'e's goal but use
program transformation and over-approximating abstract interpretation 
over a Horn clause program representation.
This allows us to apply a range of established tools and techniques 
beyond abstract interpretation, including
query-answer transformation, partial evaluation and abstraction refinement. 
We offer an iterative approach  that successively specialises 
a program.
The approach of iteratively specialising a program represented as
Horn clauses has also been
pursued by \citeN{DBLP:journals/scp/AngelisFPP14}
in order to verify program properties.
Their techniques also incorporate forward and backward propagation
of constraints, but rather than explicitly using abstract interpretation, their
specialisation algorithm involves a special constraint generalisation method.

We shall use the example in Figure~\ref{ex:backward} to demonstrate 
our approach. 
The left side shows a C program fragment, and the right its constrained 
Horn clause (CHC) representation. 
CHCs  can 
be obtained from an imperative program (containing assertions) using various approaches 
\cite{Peralta-Gallagher-Saglam-SAS98,DBLP:conf/pldi/GrebenshchikovLPR12,DBLP:conf/cav/GurfinkelKKN15,DBLP:journals/scp/AngelisFPP17}.  The A set of CHCs is not necessarily intended as an executable logic program;  in Figure~\ref{ex:backward} the predicates capture the reachable states of the computation. For example, $\mathtt{while(1,0)}$ is true if the \textbf{while} statement is reached with $a=1$ and $b=0$.  The predicate $\falsett$ represents an error state.
 Henceforth whenever we refer to a program, we refer to its CHC version. 

\begin{figure}[t]
  \begin{center}
  \begin{tabular}{ll}
    \pcode[\small]{
    $\textbf{int}~ a, b;$ \\
      \textbf{if} $(a \leq 100)$ \\
       ~ $a = 100-a$; \\
      \textbf{else} 
       $a=a-100;$ \\
      \textbf{while} ($a \geq 1 $)  \\
      	~ \{$a=a-1;$ 
      	 $b=b-2;\}$ \\
      \textbf{assert}$(b \neq 0);$
    } &
    \pcode[\small]{
$\mathtt{c1.~ init(A,B) \leftarrow true.}$ \\
$\mathtt{c2.~ if(A,B) \leftarrow A_0 \leq 100, A=100-A_0, init(A_0,B).}$\\
$\mathtt{c3.~ if(A,B) \leftarrow A_0 \geq101, A=A_0-100, init(A_0,B).}$\\
$\mathtt{c4.~ while(A,B) \leftarrow if(A,B).}$\\
$\mathtt{c5.~ while(A,B) \leftarrow A_0\geq1, A=A_0-1, B=B_0-2,}$ \\
$\mathtt{~~~~~~~~~~~~~~~~~~~~~~~~~~~while(A_0,B_0).}$\\
$\mathtt{c6.~ false\leftarrow A\leq0, B=0, while(A,B).}$
    } 
  \end{tabular}
  \end{center}
  \caption{Running example: (left) original program, (right) translation to
   CHCs}
 \label{ex:backward}
\end{figure}
For the given program, we want to ensure that $b$ is non-zero after the loop. 
The goal is to derive initial conditions on $a$ and $b$, 
sufficient to ensure that the assertion is never violated. 
The practical use of the conditions is to reject unsafe initial states 
before running the program.
We note that the assertion will 
\emph{not} be violated provided the following three conditions are met:
 (i)~if $a=100$ then $b \neq 0$, 
 (ii)~if $a<100$ then
  $2a \neq 200 - b$ and 
 (iii)~if $a>100$ then $2a \neq 200+b$.
The conjunction of these three conditions, or equivalently $b \neq
 |2a-200|$, ensures that the assertion is never violated.  
Automating the required reasoning is challenging because: 
(i)   the desired result is a disjunctive
      constraint over expressions that need an expressive domain; 
(ii)  the disjuncts cannot be represented as intervals, octagons or difference 
      bound matrices \cite{DBLP:journals/lisp/Mine06}; 
(iii) information has to be propagated forwards and backwards because we 
      lose information on $b$ and $a$ in the forward and in the backward direction respectively.
In what follows, we show how to derive the conditions automatically. 

The key contribution of this paper is 
a framework for deriving sufficient preconditions without a need
to calculate weakest preconditions or rely on abstract domains with special
properties or intricate transfer functions.
This is achieved through a combination of program transformation and
abstract interpretation, with the
derived preconditions being successively refined through iterated
transformation.

After Section~\ref{sec:prelim}'s preliminaries, we discuss, in 
Section~\ref{sec:specialisation}, the required transformation techniques.
Section \ref{sec:precond_weaker} gives iterative refinement algorithms 
that derive successively better (weaker) preconditions.
Section~\ref{sec:experiments} is an account of experimental evaluation, 
demonstrating practical feasibility of the technique.
Section \ref{sec:conclusion} concludes.

 
%

\section{Preliminaries}
\label{sec:prelim}

An atomic formula, or simply \emph{atom}, is a formula $p(\tuplevar{x})$ 
where $p$ is a predicate symbol and $\tuplevar{x}$ a tuple of arguments. 
A constrained Horn clause (CHC) is a first-order predicate logic formula 
of the form $\forall \tuplevar{x}_0 \ldots \tuplevar{x}_k (p_1(\tuplevar{x}_1)
\wedge \ldots \wedge p_k(\tuplevar{x}_k) \wedge \phi \rightarrow
p_0(\tuplevar{x}_0))$, 
where $\phi$ is a finite conjunction of quantifier-free \emph{constraints} 
on variables $\tuplevar{x}_i$ with respect to some constraint theory 
$\theory$, $p_i(\tuplevar{x}_i)$ are atoms, $p_0(\tuplevar{x}_0)$ is the 
\emph{head} of the clause and 
$p_1(\tuplevar{x}_1) \wedge \ldots \wedge p_k(\tuplevar{x}_k) \wedge \phi$ 
is the \emph{body}. 
Following the conventions of Constraint Logic Programming (CLP), 
such a clause is written as $p_0(\tuplevar{x}_0) \leftarrow \phi, 
p_1(\tuplevar{x}_1), \ldots, p_k(\tuplevar{x}_k)$. 
For concrete examples of CHCs we use Prolog-like syntax and typewriter font, 
with  capital letters for variable names and linear arithmetic constraints 
built with predicates $\mathtt{\le,\ge,<,>,=}$.
  
An \emph{Integrity constraint} is a special kind of clause whose head is the 
predicate $\mathtt{false}$. 
A \emph{constrained fact} is a clause of the form 
$p_0(\tuplevar{x}_0) \leftarrow \phi$. 
A set of CHCs is also called a program. 

Figure \ref{ex:backward} (right) contains an example of a set of 
constrained Horn clauses.  
The first five clauses define the behaviour of the program in 
Figure~\ref{ex:backward} (left) and the last clause represents a 
property of the program (that the variable $B$ is non-zero after 
executing the program) expressed as an integrity constraint.  

\paragraph{CHC semantics.} 
The semantics of CHCs is obtained using 
standard concepts from predicate logic semantics.  
An \emph{interpretation} assigns to each predicate a relation over the
domain of the constraint theory $\theory$. 
The predicate $\falsett$ is always interpreted as $\falseit$. 
We assume that $\theory$ is equipped with a decision procedure and a 
projection operator, and that it is closed under negation.
We use notation $\phi|_{V}$ to represent the constraint formulae $\phi$
projected onto variables $V$.

An interpretation \emph{satisfies} a set of formulas if each formula in 
the set evaluates to $\trueit$ in the interpretation in the standard way.
In particular, \emph{a model} of a set of CHCs is an interpretation in 
which each clause evaluates to $\trueit$. 
A set of CHCs $P$ is \emph{consistent} if and only if it has a model. 
Otherwise it is \emph{inconsistent}. 

When modelling safety properties of systems using CHCs, the consistency 
of a set of CHCs corresponds to \emph{safety} of the system.  
Thus we also refer to CHCs as being \emph{safe} or \emph{unsafe} when 
they are consistent or inconsistent respectively.

\paragraph{AND-trees and trace trees.} 
Derivations for CHCs are represented by AND-trees.
The following definitions %
are adapted from 
\citeN{Gallagher-Lafave-Dagstuhl}.

An \emph{AND-tree} for a set of CHCs is a tree whose nodes are 
labelled as follows.
\begin{enumerate}
\item
each non-leaf node corresponds to a clause (with variables suitably renamed) 
of the form $A \leftarrow \phi, A_1,\ldots,A_k$ and is labelled by the 
atom $A$ and $\phi$, and has children labelled by $A_1,\ldots,A_k$;
\item
each leaf node corresponds to a clause of the form $A \leftarrow \phi$ 
(with variables suitably renamed) and is labelled by
the atom $A$ and $\phi$; and
\item each node is labelled with the clause identifier of the 
corresponding clause.
\end{enumerate}
\noindent
Of particular interest are AND-trees with their roots labelled by the atom 
$\falseit$; these are called counterexamples.
A \emph{trace tree} is the result of removing all node labels from an 
AND-tree apart from the clause identifiers.
Given an AND-tree $t$, $\constr(t)$ represents the conjunction of the 
constraints in its node labels. 
The tree $t$ is \emph{feasible} if and only $\constr(t)$ is satisfiable 
over $\theory$. 
We also represent a conjunction of constraints as a set of constraints, 
for example, $a=0 \wedge b\geq 1$ as $\{ a=0, b\geq 1\}$. 
\begin{definition}\label{provable}
For an atom $p(\tuplevar{x})$ and a set of CHCs \(P\) we write 
$P \vdash_{\theory} p(\tuplevar{x})$ if there exists a feasible 
AND-tree with root labelled by $p(\tuplevar{x})$. 
\end{definition}

The soundness and completeness of derivation trees~\cite{JMMS} implies that 
$P$ is inconsistent if and only if $ P \vdash_{\theory} \mathtt{false}$.

\begin{wrapfigure}[13]{r}{0.22\textwidth}
\vspace{-1ex}
\small
\begin{tikzpicture}
  [auto,
   treenode/.style ={rectangle, draw, thick, text width=7.5em, text centered, 
                     rounded corners, minimum height=3em
                    },
   connector/.style={draw, thick}
  ]

  \matrix [row sep=1.0em]
  {
  \node [treenode] (c6) 
    {\begin{tabular}{c}\texttt{c6: false},\\$\mathtt{A} \leq 0 \land \mathtt{B} = 0$\end{tabular}};
  \\ \node [treenode] (c4) 
    {\begin{tabular}{c}\texttt{c4: while(A,B)},\\$\trueit$\end{tabular}};
  \\ \node [treenode] (c2) 
    {\begin{tabular}{c}\texttt{c2: if(A,B)},\\$\mathtt{A} \leq 100 \land$ \\$ \mathtt{A} = 100 - \mathtt{C}$\end{tabular}};
  \\ \node [treenode] (c1) 
    {\begin{tabular}{c}\texttt{c1: init(C,B)},\\$\trueit$\end{tabular}};
  \\
  };
  \begin{scope}[every path/.style=connector]
    \path (c6) -- (c4);
    \path (c4) -- (c2);
    \path (c2) -- (c1);
  \end{scope}
\end{tikzpicture}
\end{wrapfigure}
On the right is  an AND-tree corresponding to the derivations of 
$\mathtt{false}$ using the clauses \texttt{c6} followed by \texttt{c4, c2} 
and \texttt{c1} from the program in Figure \ref{ex:backward} (right).

\begin{definition}[Initial clauses and nodes]\label{def:initial-node}
Let $P$ be a set of CHCs, with a
{distinguished} predicate 
$p^{I}$ in $P$ which we call the \emph{initial predicate}.
The \emph{constrained facts}
$\{(p^I(\tuplevar{x}) \leftarrow \theta) 
\mid (p^I(\tuplevar{x}) \leftarrow \theta) \in P\}$ are called
the \emph{initial clauses} of $P$.
Let $t$ be an AND-tree for $P$.
A node labelled by an identifier of the clause 
$p^I(\tuplevar{x}) \leftarrow \theta$ is an \emph{initial node} of $t$.
We extend the term ``initial predicate" and use the symbol $p^I$ to 
refer also to renamed versions of the initial predicate that arise 
during clause transformations.
\end{definition}

 
%
\section{Precondition Inference}\label{sec:precond}

This section describes an approach to precondition generation. 
We limit our attention to sets of clauses for which every AND-tree for 
$\falsett$ (whether feasible or infeasible) has at least one initial node. 
Although it is not decidable for an arbitrary set of CHCs $P$ whether every 
derivation of $\falsett$ uses the initial predicate, 
the above condition on AND-trees
can be checked syntactically from the predicate dependency graph for $P$. 

\begin{definition}[Safe precondition]\label{def:precond}

Let $P$ be a set of CHCs.
Let $\phi$ be a constraint over $\theory$, and let $P'$ be the set of 
clauses obtained from $P$ by replacing the initial clauses 
$\{(p^I(\tuplevar{x}) \leftarrow \theta_i) \mid 1 \le i \le k\}$ by 
$\{(p^I(\tuplevar{x}) \leftarrow \theta_i \wedge \phi) \mid 1 \le i \le k\}$. 
Then $\phi$ is a \emph{safe precondition} for $P$ if 
$P' \not\vdash_{\theory} \falsett$.  
 
\end{definition}

Thus a safe precondition is a constraint that, 
when conjoined with the constraints on the initial predicate,
is sufficient to block derivations of $\falsett$ (given that we assume 
clauses for which $p^I$ is essential for any derivation of $\falsett$).

Ideally we would like to find the most general, or \emph{weakest} safe precondition.  
It is not computable in general, so we aim to find a condition that is as weak as possible.
The constraint $\falseit$ is always a safe precondition, albeit an uninteresting one. 
On the other hand, if $ P \not\vdash_{\theory} \falsett$ then any constraint, including $\trueit$, is a safe precondition for $P$.  
 
We first show how a safe precondition can be derived from a set of clauses.

\begin{definition}[Safe precondition $\swp(P)$ extracted from a set $P$ of clauses] 
\label{def:extract-safe} 
Let $P$ be a set of clauses. The safe precondition $\swp(P)$ is defined as:
\[
\swp(P)= \neg\bigvee \{\theta \mid  (p^I(\tuplevar{x}) \leftarrow \theta) \in P\}. 
\]
\end{definition}
\noindent
$\swp(P)$ is clearly a safe precondition for $P$ since
for each initial clause $p^I(\tuplevar{x}) \leftarrow \theta$ the conjunction $\swp(P) \wedge \theta$ is $\falseit$.
This precondition trivially blocks any derivation of $\falsett$ since we assume that every derivation of $\falsett$ uses
an initial clause.
We next show how to construct a sequence 
$P_0,P_1,\ldots,P_m$ where $P = P_0$ and each element of the sequence is more specialised with respect to
derivations of $\falsett$, and as a consequence, the constraints in the initial
clauses are stronger. Applying Definition \ref{def:extract-safe} to $P_m$ thus yields a weaker
safe precondition for $P$.

 
%
\subsection{Specialisation of Clauses}
\label{sec:specialisation}

\begin{definition}[Specialisation transformation]\label{def:spec}
Let $P$ be a set of clauses, and let $A$ be an atom. 
We write $P \Longrightarrow_A P'$  for a \emph{specialisation transformation} of $P$ with respect to $A$, yielding
a set of clauses $P'$, such that the following holds.
\begin{itemize}
\item
$P \vdash_{\theory} A$ if and only if $P' \vdash_{\theory} A$; and
\item
if $(p^I(\tuplevar{x}) \leftarrow \theta) \in P$ then there exists an initial clause $(p^I(\tuplevar{x}) \leftarrow \phi) \in P'$ such that $ \models_{\theory} \phi \rightarrow \theta$.
\end{itemize}
\end{definition}
Note that a specialisation requires not only that derivations of $A$ are preserved, but also that the initial clauses are preserved and
possibly strengthened.

\begin{lemma}\label{lemma:strengthen}
Let $P \Longrightarrow_{\mathtt{false}} P'$ be a specialisation transformation with respect to $\mathtt{false}$.
Then $\models_{\theory} \swp(P) \rightarrow \swp(P')$.
\end{lemma}
\begin{proof} This follows immediately from Definitions \ref{def:extract-safe}  and \ref{def:spec}. 
\end{proof}
We now present specific transformations for CHCs that satisfy Definition \ref{def:spec}. Applying these transformations
 enables the derivation of more precise safe preconditions. These are adapted from established techniques from the literature on CLP and Horn clause verification and analysis. 

%
\subsubsection{Specialising CHCs by Partial Evaluation (PE)}
\label{pe}

Partial evaluation  \cite{Jones-Gomard-Sestoft} 
is a transformation that
specialises a program with respect to a given input. 
The ``input" for partial evaluation of a set of CHCs $P$ is a 
(set of) constrained atom(s) $A \leftarrow \theta$. 
The result of partial evaluation is a set of CHCs $P'$ preserving the 
derivations of every instance of $A$ that satisfies $\theta$, that is, 
$P' \vdash A\phi$ if and only if $P \vdash  A\phi$
whenever $\theory \models \theta\phi$. 
The partial evaluation algorithm described here is an instantiation of the
``basic algorithm'' for partial evaluation of logic programs in \citeN{gallagher:pepm93}. 

The basic algorithm can be presented as the computation of 
the limit of the increasing sequence $S_0, S_1, S_2,\ldots$, where $S_0$ is the set of input constrained atoms and for $i\ge 0$, 
$S_{i+1} = S_0 \cup \abst_{\Psi}(\unfold_P(S_i))$. The  ``unfolding rule'' $\unfold_P$ and the abstraction operation 
$\abst_{\Psi}$ are parameters of the algorithm. 
For the algorithm used in this paper, the unfolding rule $\unfold_P(S)$ takes a set of constrained facts $S$, 
and ``partially evaluates" each element of $S$,
using the following procedure. For each $(p(\tuplevar{x}) \leftarrow \theta) \in S$, first construct the set of clauses $p(\tuplevar{x}) \leftarrow \psi' \wedge B'$ where $p(\tuplevar{x}) \leftarrow \psi \wedge B$ is a clause in $P$, and $\psi' \wedge B'$ is obtained by unfolding $\psi \wedge \theta \wedge B$ by selecting atoms so long as they are
deterministic (atoms defined by a single clause) and is not a call to an initial predicate 
or a recursive predicate, and $\psi'$ is satisfiable in $\theory$. Unfolding with this rule is guaranteed to terminate;   
$\unfold_P(S)$ returns the set of constrained facts $q(\tuplevar{y}) \leftarrow \psi'\vert_{\tuplevar{y}}$ where $q(\tuplevar{y})$ is an atom in $B'$.

The abstraction operation $\abst_{\Psi}$ ensures that the sequence $S_0, S_1, \ldots$ has a finite limit. It  performs property-based abstraction \cite{GrafS97-predabs} of a set of constrained facts with respect to a finite set of properties $\Psi$ (also a finite set of constrained facts). Then $\abst_{\Psi}(S)$ is defined as follows.
$$\begin{array}{lll}
\abst_{\Psi}(S) &=& \{ \rep_{\Psi}(p(\tuplevar{x}) \leftarrow \theta) \mid (p(\tuplevar{x}) \leftarrow \theta) \in S \},
\mathit{~~where}\\
\rep_{\Psi}(p(\tuplevar{x}) \leftarrow \theta) &=& p(\tuplevar{x}) \leftarrow \bigwedge \{ \psi
\mid( p(\tuplevar{x}) \leftarrow \psi) \in \Psi,  \theory \wedge \theta \models \psi \}\\
\end{array}$$
\noindent
The effect of $\abst_{\Psi}(S)$ is to generalise each $q(\tuplevar{y}) \leftarrow \theta \in S$ to $q(\tuplevar{y}) \leftarrow \psi$, where $\psi$ is the conjunction of properties in $\Psi$ that are implied by $\theta$. Thus only a finite number of ``versions" of $q(\tuplevar{y})$ can be generated,  ensuring that the size of the sets $S_i$ is finite (at most $2^{\vert\Psi\vert}$).  
The larger $\Psi$ is, the more versions can be produced.  More versions could cause overhead without necessarily giving more specialisation; for example, several essentially identical definitions of predicates could be produced.
Thus it is important to choose $\Psi$ taking into account both 
precision and efficiency.

In the implemented algorithm, $\Psi$ consists of the following constrained facts, generated from each clause $p(\tuplevar{x}) \leftarrow \phi,p_1(\tuplevar{x}_1),\ldots,p_n(\tuplevar{x}_n) \in P$. 
\begin{itemize}
\item
For $1 \le i \le n$, $p_i(\tuplevar{x}_i) \leftarrow \phi\vert_{\tuplevar{x}_i}$ and for each $z \in \tuplevar{x}_i$, $p_i(\tuplevar{x}_i) \leftarrow \phi\vert_{\{z\}}$.
\item
$p(\tuplevar{x}) \leftarrow \phi\vert_\tuplevar{x}$ and for each $z \in \tuplevar{x}$, $p(\tuplevar{x}) \leftarrow \phi\vert_{\{z\}}$.

\end{itemize}
The first set of constrained facts distinguishes different 
call contexts, while the second set distinguishes answers. Constraints on individual variables are extracted.
This choice of $\Psi$ was found by experiment to be a good compromise between precision and efficiency, but further experiment and analysis is needed.

Each $S_i$ in the sequence gives rise to a set of clauses 
$\renameunfold_{\Psi,P}(S_i)$, which applies the unfolding rule to each element of $S_i$ and 
renames the predicates in the resulting clauses 
according to the different versions produced by $\abst_{\Psi}$. 
The predicate $\falsett$ is not renamed.  The result returned by partial evaluation is $\renameunfold_{\Psi,P}(S_k)$, where $S_k$ is the limit
of the sequence.

\newcommand{\myleftarrow}{\leftarrow}

\begin{example}
Consider the partial evaluation of the clauses in Figure \ref{ex:backward}.  
$S_0 = \{\falsett \leftarrow \trueit\}$ and $\Psi$ consists of the following 
nine constrained facts extracted from the clauses as explained above: 

\[
\left\{
\begin{array}{l}
\mathtt{if(A,B) \leftarrow A \ge 0. \ 
if(A,B) \leftarrow A \ge 1. \
init(A,B) \leftarrow A \le 100.} \\
\mathtt{init(A,B) \leftarrow A \ge 101. \ 
while(A,B) \leftarrow A \ge 0. \ 
while(A,B) \leftarrow A \ge 1.} \\
\mathtt{while(A,B) \leftarrow A \le 0 \wedge B=0. \ 
while(A,B) \leftarrow A \le 0. \ 
while(A,B) \leftarrow B=0.}
\end{array}
\right\}
\]

Partial evaluation of the clauses generates the clauses $R_0,R_1,\ldots$ and 
sets of constrained facts $S_0, S_1,\ldots$ as shown in Figure \ref{run_pe}.

Note that three versions of the $\mathtt{init}$ predicate are generated 
(from the new constrained facts generated in steps 3 and 4), 
each having different constraints. As we will see in the next section, 
this allows the extraction of more precise preconditions for safety of 
the clauses than could be obtained from the original clauses. 

\end{example}

\begin{figure}[t]
\begin{tabular}{|l|l|l|}
 \cline{1-3}
$i$ & $S_i$ & \textbf{$R_i=\renameunfold_{\Psi,P}(S_i)$} \\ \cline{1-3} \cline{1-3}
0  & $\begin{aligned} S_0=\{ \falsett \leftarrow \trueit\} \end{aligned}$ & $\begin{aligned}  R_0=  \{\mathtt{false \myleftarrow A \le 0,B=0,while\_7(A,B)}. \}\end{aligned}$ \\ \cline{1-3}
1  & $\begin{aligned} S_1= S_0 ~\cup \\ \{ \mathtt{while(A,B) \leftarrow A \le 0,B=0.} \} \end{aligned}$ & $\begin{aligned} R_1= R_0 ~\cup  \\ \{ \mathtt{while\_7(A,B) \myleftarrow A \le 0,B=0, if\_6(A,B).} \\  \mathtt{while\_7(A,B) \myleftarrow A=0,B=0,}  \\  \mathtt{C=1,D=2,while\_5(C,D)}.\}  \end{aligned}$ \\ \cline{1-3}

2  & $\begin{aligned} & S_2= S_1~ \cup \\ & \{\mathtt{while(A,B) \leftarrow A\ge 1.} \\ & \mathtt{if(A,B) \leftarrow \trueit .} \}\end{aligned}$ & $\begin{aligned}&R_2= R_1 \cup  \{ \mathtt{while\_5(A,B) \myleftarrow  A\ge 1,if\_2(A,B).} \\ & 
\mathtt{while\_5(A,B) \myleftarrow  A\ge 1,C-A= 1,D-B=2,} \\ & ~~~~~~~~~~~~~~~~~~~~~~\mathtt{while\_5(C,D)}. \\ &
\mathtt{if\_6(A,B) \myleftarrow  A\ge 0,A+C=100,init\_4(C,B).} \\ &
\mathtt{if\_6(A,B) \myleftarrow  A\ge 1,C-A= 100,init\_3(C,B).} \}
 \end{aligned}$ \\ \cline{1-3}
3  &  $\begin{aligned} & S_3= S_2 \cup \{ \mathtt{if(A,B) \leftarrow A\ge 1.} \\ & \mathtt{init(A,B) \leftarrow A\ge 101.} \\ & \mathtt{init(A,B) \leftarrow A\leq  100.} \}
 \end{aligned}$ & $\begin{aligned}& R_3= R_2 ~\cup \\ &
 \{ \mathtt{if\_2(A,B) \myleftarrow  A\ge 1,A+C=100,init\_1(C,B).} \\ &
\mathtt{if\_2(A,B) \myleftarrow  A\ge 1,C-A= 100,init\_3(C,B).} \\ &
\mathtt{init\_4(A,B) \myleftarrow  A\leq  100. } \\ &
\mathtt{init\_3(A,B) \myleftarrow  A\ge 101.} \}
\end{aligned}$ \\ \cline{1-3}
 
4 & $\begin{aligned} S_4= S_3 ~\cup \\  \{ \mathtt{init(A,B) \leftarrow  A\leq  99.} \} \end{aligned}$ & $\begin{aligned} & R_4= R_3 \cup 
\{ \mathtt{init\_1(A,B) \myleftarrow  A\leq 99.} \}
 \end{aligned}$ \\ \cline{1-3}

5 & $S_5= S_4$ & $R_5= R_4 $ \\ \cline{1-3}
\end{tabular}
\caption{Steps performed during the run of  partial evaluation}
\label{run_pe}
\vspace{1em}
\end{figure}

 \begin{lemma}
 \label{lemma:pe}
 Partial evaluation using the procedure described above is a specialisation transformation (Definition \ref{def:spec}).
 \end{lemma}
 \begin{proof}
 The algorithm satisfies the standard condition of partial evaluation that it preserves derivations of the given goal atom. 
 The strengthening of the initial clauses follows from the
 fact that our unfolding rule does not unfold the initial predicate.  Hence the result contains the initial clauses from the original,
 with constraints possibly strengthened by the call constraints in the algorithm. (If a clause is never called, its constraint is strengthened to $\falseit$).
 \end{proof}

The safe precondition of the partially evaluated clauses is $\neg(\mathtt{A \le 99} \vee \mathtt{A \le 100} \vee \mathtt{A \ge101})$, which is 
equivalent to $\falseit$ (over the integers).
Thus partial evaluation has not improved the safe precondition compared to the original clauses in Figure~\ref{ex:backward}.  However,
the splitting of the initial clauses enables a further specialisation, which is described next.

 
\subsubsection{Transforming CHCs by Constraint Specialisation (CS)}
\label{cs}
Constraint specialisation is a transformation that strengthens the constraints in a set of CHCs, while preserving
derivations of a given atom. Consider the following simple example in Figure \ref{cs_example} (left) that motivates the principles of the transformation.  

\begin{figure}[t]
\centerline{
  \begin{tabular}{ll}
    \pcode[\small]{
$\texttt{false} \myleftarrow \texttt{A} \ge \texttt{0},\texttt{p(A,B)}$. \\
$\texttt{p(A,B)} \myleftarrow \texttt{C}\ge\texttt{A},\texttt{p(C,B)}$.\\
$\texttt{p(A,B)} \myleftarrow \texttt{A}=\texttt{B}$.
    } &
    \pcode[\small]{
$\texttt{false} \myleftarrow \texttt{A} \ge\texttt{0}, \underline{\texttt{B}\ge \texttt{A},\texttt{A}\ge\texttt{0}}, \texttt{p(A,B)}$.\\
$\texttt{p(A,B)}\myleftarrow \texttt{C}\ge \texttt{A}, \underline{\texttt{B} \ge \texttt{C,C} \ge \texttt{0}}, \texttt{p(C,B)}$.\\
$\texttt{p(A,B)}\myleftarrow \texttt{A}=\texttt{B},\underline{\texttt{B}\ge \texttt{A,A}\ge\texttt{0}}$.
} 
  \end{tabular}
}
\caption{Example program (left) and its constraint specialised version (right)\label{cs_example}}
\end{figure}
\noindent
Assume we wish to preserve derivations of $\mathtt{false}$. 
The transformation in Figure~\ref{cs_example} (right) is a constraint specialisation
with respect to $\mathtt{false}$.
\noindent
The strengthened constraints are obtained by recursively propagating $\mathtt{A \ge 0}$ top-down from the goal $\mathtt{false}$
and $\mathtt{A = B}$ bottom-up from the constrained fact. An invariant $\texttt{B} \ge \texttt{A}, \texttt{A} \ge \texttt{0}$ for the derived answers of the recursive predicate $\texttt{p(A,B)}$ in derivations of $\falsett$ is computed and conjoined to each call to $\texttt{p}$ in the clauses (underlined in the clauses in Figure \ref{cs_example} (right)).

\begin{definition}[Constraint specialisation]\label{def:cs}
A constraint specialisation of $P$ with respect to a goal $A$ is a transformation in which 
 each constraint $\phi$ in a clause of $P$ is replaced by a constraint $\psi$ where $ \models_{\theory} \psi \rightarrow \phi$,
 such that the resulting set of clauses is a specialisation transformation (Definition \ref{def:spec}) of $P$ with respect to A.
 \end{definition}

 \noindent 
 In our experiments, the combined top-down and bottom-up propagation of constraints illustrated above is achieved by
 abstract interpretation over the domain of convex polyhedra applied to a query-answer transformed version of the 
 set of CHCs.  The method is described in detail in \citeN{DBLP:journals/scp/KafleG17}.  The result of
 applying constraint specialisation to the output of partial evaluation of the running example is shown in Figure \ref{pe_cs_1}.  Note that the second clause 
 for $\mathtt{if\_6}$ has been eliminated, since its constraint was specialised to $\falseit$.
 
 The safe precondition derived after constraint specialisation from the 
initial clauses in Figure~\ref{pe_cs_1} is 
\[
  \neg(\mathtt{(A=100}\wedge \mathtt{B=0)} \vee \mathtt{(A}\le\mathtt{99} \wedge \mathtt{2A+B=200)} \vee 
\mathtt{(A}\ge\mathtt{101} \wedge \mathtt{2A-B=200))}
\]
This simplifies (over the integers) to $\mathtt{B\neq \vert 2A-200 \vert}$,
which is the condition obtained in Section \ref{sec:intro} and is optimal 
(weakest).  
\begin{figure}
\vspace*{2ex}
\pcode[\small]{
$\mathtt{false \leftarrow A=0,B=0,while\_7(A,B).}$ ~~$\mathtt{while\_7(A,B) \leftarrow  A=0, B=0, if\_6(A,B).}$\\
$\mathtt{while\_7(A,B) \leftarrow A=0, B=0,C=1, D=2, while\_5(C,D).}$\\
$\mathtt{if\_6(A,B) \leftarrow A=0,  B=0, C=100,  init\_4(C,B).}$ \\
$\mathtt{while\_5(A,B) \leftarrow  A \geq 1,  2A-B=0, if\_2(A,B).}$\\
$\mathtt{while\_5(A,B) \leftarrow  A \geq 1,  2A=B,C-A=1,  D-2A=2,  while\_5(C,D).}$\\
$\mathtt{if\_2(A,B) \leftarrow  A \geq 1, 2A=B,A+C=100, init\_1(C,B).  }$\\
$\mathtt{if\_2(A,B) \leftarrow  A \geq 1,  2A=B, C-A=100,  init\_3(C,B).}$\\
$\mathtt{init\_4(A,B) \leftarrow  A=100, B=0.}$\\
$\mathtt{init\_3(A,B) \leftarrow  A \geq 101, 2A-B=200.}$\\
$\mathtt{init\_1(A,B) \leftarrow A \leq  99,2A+B=200.}$
} 
\caption{Constraint specialisation of the partially evaluated clauses in Figure~\ref{run_pe}\label{pe_cs_1}}
\end{figure}

\subsubsection{Transforming CHCs by Trace Elimination (TE)}
\label{te}

Let $P$ be a set of CHCs and let $t$ be an AND-tree for $P$. It is possible to construct a set of clauses $P'$ 
which preserves the set of AND-trees (modulo predicate renaming) of $P$, apart from $t$. 
The transformation from $P$ to $P'$ is called \emph{trace elimination} (of $t$). We have previously described a
technique for trace elimination \cite{DBLP:journals/cl/KafleG17}, based on the difference operation on finite tree automata.
In that work, trace elimination played the role of a refinement operation, in which infeasible traces 
were removed from a set of CHCs in a counterexample-guided verification algorithm in the CEGAR style \cite{DBLP:journals/jacm/ClarkeGJLV03}.

For the purpose of deriving safe preconditions of a set of clauses $P$, we apply trace elimination 
 to eliminate both infeasible and feasible AND-trees.   AND-trees for $\falsett$ are obtained naturally from transformations such as partial evaluation or constraint specialisation.
 First consider the elimination of an infeasible AND-tree.

\begin{lemma}
\label{lemma:te-feasible}
Let $P'$ be the result of eliminating an \emph{infeasible} AND-tree $t$ for $\mathtt{false}$ from $P$. Then
$P \Longrightarrow_{\mathtt{false}} P'$.
\end{lemma}
\begin{proof}
All derivations of $\mathtt{false}$ are preserved, and the transformation 
generates only predicate-renamed copies of the 
original clauses, hence the initial clauses are preserved.
\end{proof}
\noindent
So in this case we have $\models_{\theory} \swp(P) \rightarrow \swp(P')$. 
However, the elimination of a feasible AND-tree $t$ for $\mathtt{false}$ 
is not as straightforward.   
Nevertheless, we can still
use this transformation to derive safe preconditions, by the following lemma. 
\begin{lemma}\label{lemma:feasible}
Let $P'$ be the result of eliminating a \emph{feasible} AND-tree $t$ for $\mathtt{false}$ from $P$. 
Let $p^I(\tuplevar{x})$ be the atom label of an initial node of $t$ and let $\theta = \constr(t)\vert_\tuplevar{x}$.
Then $\swp(P) = \swp(P') \wedge \neg\theta$.
\end{lemma}

\begin{proof}
$\neg\theta$ is a sufficient condition, when conjoined with the body of the clause labelling the initial node,
to make $t$ infeasible.  All other derivations of $\falsett$ from $P$ are preserved in $P'$.  Hence the conjunction of
$\neg\theta$ and $\swp(P')$ is a safe precondition for $P$.
\end{proof}

The usefulness of trace elimination is twofold.  Firstly, it can cause splitting of the initial predicates, resulting in disjunctive pre-conditions.
Secondly, the elimination of a feasible trace acts as a decomposition of the problem.

 
%
\subsection{Inferring Weaker Preconditions}
\label{sec:precond_weaker}

We can combine the various transformations to derive weaker preconditions, as shown in the following two propositions.
\begin{proposition}\label{prop:spec-seq}
Let $P = P_0$ and let the sequence $P_0, P_1,\ldots,P_m$ be a sequence such that $P_i \Longrightarrow_{\falsett} P_{i+1}$ $(0 \le i < m)$.
Then $ \models_{\theory} \swp(P) \rightarrow \swp(P_m)$.
\end{proposition}
\begin{proof}
By induction on the length of the sequence, applying Lemma \ref{lemma:strengthen}.
\end{proof}

 \noindent
If we also eliminate feasible traces, then we have to keep track of the substitutions arising from the eliminated trees.

\begin{proposition}\label{prop:sequence}
Let $P = P_0$, $\psi_0 = \trueit$ and let the sequence $(P_0, \psi_0),(P_1, \psi_1 ),\ldots,(P_m, \psi_m)$ be a sequence of pairs where 
for $(0 \le i < m)$
\begin{itemize}
\item
either $P_i \Longrightarrow_{\falsett} P_{i+1}$ and $\psi_{i} = \psi_{i+1}$, or 
\item
$P_{i+1}$ is obtained by eliminating a feasible trace $t$ from $P_i$ , and
$\psi_{i+1} = \psi_i \wedge \neg\theta$, where $\neg\theta$ is the constraint extracted from $t$, as in Lemma \ref{lemma:feasible}.
\end{itemize}
Then $ \models_{\theory} \swp(P) \rightarrow (\swp(P_m) \wedge \psi_m)$.
\end{proposition}
\begin{proof}
By induction on the length of the sequence, applying Lemma \ref{lemma:strengthen} and Lemma \ref{lemma:feasible}.
\end{proof}
\noindent
Proposition \ref{prop:sequence} establishes the correctness of the algorithm 
used in Section \ref{sec:experiments}, and any other algorithm that applies 
partial evaluation, constraint specialisation and trace elimination 
in any order.
Proposition \ref{prop:spec-seq} is a special case of
Proposition \ref{prop:sequence}: if we do not eliminate any feasible trees
then $\psi_m$ is $\trueit$ and so 
$\models_{\theory} \swp(P) \rightarrow \swp(P_m)$.

As we have shown, applying partial evaluation followed by constraint 
specialisation for our running example was sufficient to derive the 
weakest safe precondition. 
However, in more complex cases we need one or more iterations of these 
operations, possibly with the elimination of feasible AND-trees as well. 
In Figure \ref{example_t4} we show an example taken from \citeN{DBLP:conf/pldi/BeyerHMR07} in which 
repeated application of partial evaluation followed by constraint 
specialisation does not achieve a useful result, but where the elimination 
of a single feasible AND-tree causes an optimal precondition to be generated. 
The optimal precondition for this program is $\mathtt{init(I,A,B,N) \leftarrow N\leq I \wedge A+B=3*N}$.
To derive this, one needs to propagate constraints from the third and the fourth clauses (constrained facts corresponding to the predicate \texttt{l}) to the \texttt{init} clause. Since these constraints are disjunctive (arising from two different clauses), the propagation should be able to split the \texttt{init} predicate. PE can often perform splitting but not in this case since the recursive predicate \texttt{l} is not unfolded, owing to the potential for a resulting blowup. 
\begin{figure}[t]
\pcode[\small]{
$\mathtt{false\leftarrow init(I,A,B,N), l(I,A,B,N).}$\\
$\mathtt{l(I,A,B,N)\leftarrow
    I < N,
    l\_body(A,B,A1,B1),
    I1 = I+1,
    l(I1,A1,B1,N).}$\\
$\mathtt{l(I,A,B,N)\leftarrow
    I \geq N, A + B > 3 * N.}$\\
$\mathtt{l(I,A,B,N)\leftarrow
    I \geq N, A + B < 3 * N. }$\\       
$\mathtt{l\_body(A0,B0,A1,B1)\leftarrow A1 = A0+1, B1 = B0+2.}$\\
$\mathtt{l\_body(A0,B0,A1,B1)\leftarrow A1 = A0+2, B1 = B0+1.}$\\
$\mathtt{init(I,A,B,N).}$
}
\caption{Example requiring trace elimination\label{example_t4}}
\end{figure}

We now show how trace-elimination together with other transformations 
allows us to derive this condition. 
Applying CS followed by PE to Figure~\ref{example_t4} gives us the program in 
Figure~\ref{example_t4_cs0} 
(we have labelled the clauses for the purpose of presentation). 
If we derive a precondition from this program, we will get trivial $\falseit$. 
As a next step, we search for a derivation (counterexample) violating safety.
The trace tree \texttt{c1(c10,c2(c8,c5(c8,c5(c8,c5(c8,c6))))))} 
(using its term representation) is a feasible counterexample. 
Then we remove this from the program in Figure~\ref{example_t4_cs0} 
using the automata-theoretic approach described by 
\citeN{DBLP:journals/cl/KafleG17}. 
In summary, the approach consists of representing the program as well as 
the trace to be removed as finite tree automata, performing automata 
difference and generating a new program from the difference automaton. 
The new program is guaranteed not to contain the particular trace any more.

\begin{figure}[t]
\vspace*{3ex}
\pcode[\small]{
$\mathtt{c1.~~~ false \leftarrow
   init(A,B,C,D),
   l\_3(A,B,C,D).}$\\
$\mathtt{c2.~~~ l\_3(A,B,C,D) \leftarrow
   -C+F>=1,
   -A+D>0,
   C-F>= -2, 
   A-E= -1,}$ \\
   \> \>$ \mathtt{B+C-F-G= -3,
   l\_body\_2(B,C,G,F),
   l\_1(E,G,F,D).}$\\

$\mathtt{c3.~~~ l\_3(A,B,C,D) \leftarrow
   B+C-3*D>0,
   A-D>=0.}$\\

$\mathtt{c4.~~~ l\_3(A,B,C,D) \leftarrow
   -B-C+3*D>0,
   A-D>=0.}$\\

$\mathtt{c5.~~~ l\_1(A,B,C,D) \leftarrow
   -C+F>=1,
   -A+D>0,
   C-F>= -2,
   A-E= -1,}$ \\
   \> \> $\mathtt{B+C-F-G= -3,
   l\_body\_2(B,C,G,F),
   l\_1(E,G,F,D).}$\\

$\mathtt{c6.~~~ l\_1(A,B,C,D) \leftarrow
   B+C-3*D>0,
   -A+D> -1,
   A-D>=0.}$\\

$\mathtt{c7.~~~ l\_1(A,B,C,D) \leftarrow
   -B-C+3*D>0,
   -A+D> -1,
   A-D>=0.}$\\

$\mathtt{c8. ~~~ l\_body\_2(A,B,C,D) \leftarrow
   A-C= -1,
   B-D= -2.}$\\

$\mathtt{c9.~~~ l\_body\_2(A,B,C,D) \leftarrow
   A-C= -2,
   B-D= -1.
  }$\\

$\mathtt{c10.~  init(A,B,C,D).}$
}
\caption{The constraint specialisation of the program in Figure~\ref{example_t4}%
\label{example_t4_cs0}
}
\end{figure}

The removal causes the splitting of the predicate \texttt{l}, which the partial evaluation can take advantage of in the next iteration. Re-application of PE followed by CS generates the following clauses for \texttt{init} predicates (other clauses are not shown).

\begin{center}
\pcode[\small]{
$\mathtt{init\_1(A,C,D,B)\leftarrow B>A}.$\\
$\mathtt{init\_2(A,C,D,B)\leftarrow  A>=B,
   C+D>3B.}$\\
$\mathtt{init\_3(A,C,D,B)\leftarrow A>=B,
   3*B>C+D.}$
}
\end{center}

\noindent
Then the derived safe precondition is:
\[
\mathtt{init(A,C,D,B) \leftarrow \neg( (B>A) \vee(A\geq B \wedge C+D>3B) \vee (A \geq B \wedge
   3*B>C+D)).}
\]
Simplifying the formula and mapping to the original variables, 
we obtain the following formula as the final precondition 
\[
  \mathtt{init(I,A,B,N) \leftarrow N\leq I \wedge A+B=3*N}.
\]
There is, however, a performance-precision trade-off when removing 
(in)feasible AND-trees. 
Trace elimination helps derive precise preconditions at the cost of 
performance; the \emph{Fischer} protocol is an example of this. 
It requires 4 iterations of PE followed by CS to generate the optimal 
precondition (obtained in $\approx$8 seconds), whereas these iterations 
interleaved by trace elimination require only 3 iterations 
(but obtained in $\approx$35 seconds). 

 
%
\section{Experimental Evaluation}
\label{sec:experiments}

\subsection{Benchmarks}
We have experimented with three kinds of benchmarks. 
\begin{enumerate}
\item \emph{Unsafe I}:
Examples that are known to be unsafe, where the initial states are over-general.
In such cases the aim of safe precondition generation is to find out whether there is
a useful subset of the initial states that is safe. 
\item \emph{Unsafe II}:
Examples that are known to be unsafe, where the initial state is a counterexample state from which $\falsett$
can be derived.  In this case it is pointless to try to find a safe subset as above, so we remove the given
constraint on the initial state, and then try to derive a non-trivial safe precondition.
\item \emph{Safe}:
Examples that are safe for given initial states.  
In such cases, our aim is to try to weaken the conditions on the initial states. This is done by
removing the given constraints from the initial states and then deriving safe preconditions.
If we can generate safe preconditions that are more general than the original  constraints
then we have generalised the program without losing safety.
\end{enumerate}
For the experiments, we collected a set of 241 (188 safe/53 unsafe) programs from a variety of sources.  Most are from the repositories of 
state-of-the-art software verification tools such as DAGGER\footnote{\url{http://www.cfdvs.iitb.ac.in/~bhargav/dagger.php}}
\cite{DBLP:conf/tacas/GulavaniCNR08}, TRACER\footnote{\url{https://github.com/tracer-x/tracer/tree/master/test/transformation}}
\cite{DBLP:conf/cav/JaffarMNS12}, InvGen\footnote{\url{http://www.mpi-sws.org/~agupta/invgen}} \cite{DBLP:conf/cav/GuptaR09}, and
from the TACAS 2013 Software Verification Competition 
\cite[Control flow and Loops categories]{DBLP:conf/tacas/Beyer13}.
\footnote{Translated to CHCs using the program specialisation approach of 
\citeN{DBLP:journals/scp/AngelisFPP17}.}
Other examples are from the literature on  precondition generation, 
backwards analysis and parameter synthesis
\cite{DBLP:conf/sas/BakhirkinBP14,DBLP:journals/entcs/Mine12,mine-backward-analyzer,DBLP:conf/vmcai/Moy08,DBLP:conf/sas/BakhirkinM17,DBLP:conf/rp/CassezJL17}
and manually translated to CHCs. 
These benchmarks are designed to demonstrate/test the strengths/usability 
of different tools and methods proposed to solve software verification, 
parameter synthesis and precondition generation problems and contain up 
to approximately 500 lines of code.
Finally there are examples crafted by us; these are simple but non-trivial 
examples whose optimal precondition can be derived manually.

\subsection{Implementation} 
We implemented an algorithm that builds a sequence as defined in 
Proposition~\ref{prop:sequence}, of length $3n+2$ ($n \ge 0$), 
iteratively applying the transformations $pe$ (partial evaluation), 
$cs$ (constraint specialisation) and $te$ (trace elimination). 
The safe precondition for $P$ is 
$\swp( cs~ \circ ~pe~ \circ~ (te~ \circ ~cs ~\circ ~pe)^n(P))$ $(n \ge 0)$.
This particular sequence of transformations is based on the rationale that 
constraint specialisation is most effective when performed just after partial
evaluation, which propagates constraints and introduces new versions
of predicates. 
Trace elimination is more expensive and is performed only after the first 
iteration. 
In future work we will experiment with other strategies, especially
to limit the application of $te$.
The implementation is based on components from
the \rahft\ verifier \cite{DBLP:conf/cav/KafleGM16}. This accepts CHCs (over
the background theory of linear arithmetic) as input and returns
a Boolean combination of linear constraints in terms of the initial state
variables as a precondition. 
The tool is written in Ciao Prolog
\cite{Ciao} and uses Yices 2.2 \cite{Dutertre:cav2014} and
the Parma Polyhedra Library \cite{BagnaraHZ08SCP} for constraint manipulation.
The experiments were carried
out on a MacBook Pro with a 2.7 GHz Intel Core i5 processor and 16 GB
memory running OS X 10.11.6, with a timeout of 300 seconds for each example.

\begin{table}[t]
    \begin{tabular}{|l|r|r|r|r|}
    \cline{1-5}
    ~            & \textrm{$n=0$} & \textrm{$n=1$} & \textrm{$n=2$} & \textrm{$n=3$} \\ \cline{1-5} \cline{1-5}
   \multicolumn{5}{c}{\textrm{Safe instances (188)}} \\ \cline{1-5}
    \textrm{non-trivial (more general)}     & 119  (101)                      & 143 (125)                           & 156                           (129) & \textrm{160}  (131)         \\ \cline{1-5}
    \textrm{trivial/timeouts}     & 69/0                        & 45/3                            & 32/10                            & 28/16           \\ \cline{1-5}
    \textrm{avg. time (sec.)}  & \textrm{1.45}        & 14.69                      & 27.52                         &36.73          \\ \cline{1-5}
    \multicolumn{5}{c}{\textrm{Unsafe   I  instances (17)}} \\ \cline{1-5}
    \textrm{non-trivial}     & 16                        & \textrm{17}                            & \textrm{17}                            & \textrm{17}           \\ \cline{1-5}
    \textrm{trivial/timeouts}     & 1/0                        & 0/0                            & 0/0                            & 0/0           \\ \cline{1-5}
    \textrm{avg. time (sec.)}  & \textrm{0.23}        & 0.82                      & 1.64   &      3.35                           \\ \cline{1-5}
        \multicolumn{5}{c}{\textrm{Unsafe   II instances (36)}} \\ \cline{1-5}
    \textrm{non-trivial}     & 9                        & \textrm{12}                            & \textrm{12}                            & \textrm{12}           \\ \cline{1-5}
    \textrm{trivial/timeouts}     & 27/0                        & 24/2                            & 24/7                            & 24/7           \\ \cline{1-5}
    \textrm{avg. time (sec.)}  & \textrm{3.38}        & 50.41                      & 64.72   &     70.91                           \\ \cline{1-5}
    \end{tabular}

     \caption{Results on 241 (188 safe and 53 unsafe) programs; timeout 5 minutes\label{tbl:experiments}}
\end{table}

\subsection{Discussion} 
Experimental results are shown in Table~\ref{tbl:experiments},
for varying number of specialisation iterations $n$.
The classifications ``more general" and ``non-trivial" in Table \ref{tbl:experiments} relate the derived precondition $I$ with the original condition on the initial states $O$.
If $ \models_{\theory} I \not\equiv \falseit$ then the result is non-trivial.  If $\models_{\theory} O \rightarrow I$ then the derived precondition 
is more general than the given initial states. For the \emph{safe} benchmarks, the ``more general" results are a subset of the ``non-trivial" results,
while for the \emph{unsafe} benchmarks, the result cannot be more general than the original (unsafe) condition and so there are no ``more general" results. \emph{Timeouts} indicates the number of timeouts in the current iteration.  When there is a timeout in the current iteration, the precondition is the precondition generated in the previous iteration. Therefore, the timeouts in the current iteration correspond to trivial, non-trivial or timeouts in the previous iteration. Thus, the trivial instances in the current iteration is the sum of trivial instances in this iteration and the trivial instances in the previous iteration of the current timeouts.  

The choice of 3 iterations is motivated by the following observations 
(though we can stop at any iteration and still derive a precondition): 
(i) for the categories literature and hand-crafted benchmarks, 3 iterations 
suffice to reproduce earlier results, and 
(ii) iterations beyond the third yield negligible improvements but more
timeouts.

For the \emph{safe} benchmarks, the algorithm succeeds for $n=3$ in generalising the safe initial conditions in 131 of the 188 benchmarks, and returns a non-trivial safe precondition in 160 of them. The remainder either return trivial results or a timeout.
A higher proportion of the \emph{unsafe} benchmarks return a trivial safe precondition,
even when the initial state constraints are removed.
A possible reason is that some of these unsafe programs are designed with an internal bug, and
thus have no safe initial states. If the analysis returns a trivial safe precondition, it
could be due to imprecision of the analysis, but could also be an indication to the programmer
to look for the problem elsewhere than in the initial states.

The results in the column $n=0$ show that the specialisation 
$(cs ~\circ ~pe)$ alone can infer non-trivial preconditions for a large 
number of benchmarks, namely $63\%$ (safe) and $37\%$ (unsafe)
instances both in less than 10 seconds. 
Among 119 non-trivial safe instances, 101 are generalised constraints. 
 
Further specialisation ($n>0$) increases the number of non-trivial and 
generalised preconditions by relatively small percentages of the total. 
The increased precision of the
preconditions comes at a significant cost in time.   
For Safe, Unsafe I, and Unsafe II instances, the average time 
goes from 1.45, 0.23 and 3.38 seconds, respectively, when $n=0$, to
36.73, 3.35 and 70.91 seconds, when $n=3$.
However, our prototype implementation is amenable to much optimisation, 
including sharing results from one iteration to the next,
which could reduce the overhead. 

For the categories of literature and hand-crafted benchmarks in which we 
know the weakest safe precondition, the tool is able to reproduce the 
results from the literature, see Figure \ref{tbl:examples}. The results were generated in at most 1 iteration in less than a second, except for the \emph{Fischer protocol}, which required 3 iterations and 35 seconds. 
\begin{figure}[t]
    \begin{tabular}{|l|l|}
    \cline{1-2}
    \textbf{Program} & \textbf{Precondition} \\  \cline{1-2}
    bakhirkin-fig3 \cite{DBLP:conf/sas/BakhirkinBP14}      & \(\displaystyle (1 \leq a \leq 99 \rightarrow b \geq 1) ~\wedge ~ (a \leq 0 \rightarrow b \neq 0) \)       \\  \cline{1-2}
    bakhirkin  \cite{DBLP:conf/sas/BakhirkinBP14}     & \(\displaystyle   1 \leq a \leq 60 \vee a \geq 100 \)           \\  \cline{1-2}
    mine \cite{DBLP:journals/entcs/Mine12}      & \(\displaystyle  0 \leq a \leq 5 \)          \\ \cline{1-2}
    mon\_fig1  \cite{DBLP:conf/sas/BakhirkinM17}     & \(\displaystyle   a=b ~\wedge~ a \geq 0 \)             \\  \cline{1-2}
    moy    \cite{DBLP:conf/vmcai/Moy08}   & \(\displaystyle  b < 1 ~ \vee ~( b < 2 \wedge a > 0)  \)           \\  \cline{1-2}
    navas2   (crafted)    & \(\displaystyle  a \leq 99 ~\vee ~ b \geq 100  \)            \\  \cline{1-2}
    simple\_function   \cite{mine-backward-analyzer}    & \(\displaystyle   6 \leq a \leq 61 \)            \\  \cline{1-2}
    test\_both\_branches \cite{mine-backward-analyzer}     & \(\displaystyle   3 \leq a \leq 17\)            \\  \cline{1-2}
    test\_nondet\_body  \cite{mine-backward-analyzer}     & \(\displaystyle   6 \leq a \leq 13\)           \\  \cline{1-2}
    test\_nondet\_cond  \cite{mine-backward-analyzer}     & \(\displaystyle   3 \leq a \leq 17\)            \\  \cline{1-2}
    test\_then\_branch   \cite{mine-backward-analyzer}    & \(\displaystyle   10 \leq a \leq 20\)            \\  \cline{1-2}
        fischer   \cite{DBLP:conf/rp/CassezJL17}    & \(\displaystyle   a+2c<b \vee a<0 \vee b< 0 \vee c \leq 0\)            \\  \cline{1-2}
Jhala   \cite{DBLP:conf/tacas/JhalaM06}    & \(\displaystyle   a<0 \vee a \geq b \vee c \neq d\)            \\  \cline{1-2}
Ball SLAM  \cite{DBLP:conf/ifm/BallCLR04}    & \(\displaystyle  b < c\)            \\  \cline{1-2}
client ssh protocol    & \(\displaystyle   b<a \vee b<2 \vee a>3\)            \\  \cline{1-2}
\citeN{DBLP:conf/pldi/BeyerHMR07}    & \(\displaystyle   n\leq i \wedge a+b=3n\)            \\  \cline{1-2}
    \end{tabular}
   
    \caption{Examples and their safe preconditions\label{tbl:examples}}
\end{figure}
As well as reproducing challenging examples from the literature 
(Figure~\ref{tbl:examples}), we are able to apply the technique 
to larger examples (shown in Table~\ref{tbl:experiments}) than have 
previously been dealt with by automatic methods for precondition generation;
we are also able to solve challenging examples that were not solvable by 
previous automatic techniques 
(such as our running example from Figure~\ref{ex:backward}).


%
\section{Concluding Remarks}
\label{sec:conclusion}
We have presented a framework for computing a sufficient 
precondition of a program with respect to assertions; it
enables derivation of more precise preconditions through iteration.
Rather than relying on weakest precondition calculation or intricate 
transfer functions, it uses off-the-shelf components from program 
transformation and abstract interpretation, which eases implementation.
Furthermore, the approach does not depend on specific abstract domain properties such as
pseudo-complementation but is domain-independent and generic. 
By this we mean that the individual specialisation transformations
such as partial evaluation and constraint specialisation 
can be adapted to different abstract domains with their usual 
precision/performance limits, while still using features of the framework such as 
iteration and disjunctive constraints that arise from polyvariant specialisation.
Evaluation on a set of benchmarks is promising.  
We are currently investigating the conditions under which the derived 
preconditions are the weakest possible, as well as other improved termination 
criteria for refinement with the aim of generating optimal preconditions.

 
\section*{Acknowledgements}
We are grateful for support from the Australian Research Council.
The work has been supported by Discovery Project grant DP140102194, and
Graeme Gange is supported through Discovery Early Career Researcher
Award DE160100568.
We wish to thank Jorge Navas, for useful discussions based on an 
early draft of the manuscript, Emanuele De Angelis, for making 
benchmarks available to us, and the three anonymous reviewers, 
for suggestions which led to clear improvements of the paper.

\bibliographystyle{acmtrans}
\end{document}